\documentclass[letterpaper, 10 pt, conference]{ieeeconf}  % Comment this line out if you need a4paper

\IEEEoverridecommandlockouts                              % This command is only needed if 
\usepackage{graphics} % for pdf, bitmapped graphics files
\usepackage{amsmath} % assumes amsmath package installed
\usepackage{amssymb}  % assumes amsmath package installed

\newtheorem{theorem}{Theorem}
\newtheorem{lemma}{Lemma}

\newtheorem{remark}{Remark}

\newtheorem{assumption}{Assumption}
\newcommand{\tighten}{\text{tighten}}
\newcommand{\red}[1]{\textcolor{black}{#1}}
\newcommand{\blue}[1]{\textcolor{black}{#1}}

\usepackage{cite}
\usepackage{subcaption}
\usepackage{optidef}
\usepackage{color,comment}
\usepackage{hhline}
\usepackage{makecell}
\usepackage{bm}
\title{\LARGE \bf
Formal Synthesis of Robust Koopman-Model Predictive Control: \\  A Case Study in AC-DC Power Conversion
}
\author{%
Shun Hirose, Shiu Mochiyama, and Yoshihiko Susuki
\thanks{This work was supported in part by JSPS KAKENHI (Grant No. 23H01434 and No. 26K00968) and JSPS Bilateral Collaborations (Grant No. JPJSBP120242202). }
\thanks{They are with Department of Electrical, Electronic and Digital Science and Engineering, Kyoto University, Katsura, Nishikyo-ku, Kyoto 615-8510 Japan {\tt hirose.shun.65x@st.kyoto-u.ac.jp}, {\tt \{mochiyama.shiu, susuki.yoshihiko\}.5c@kyoto-u.ac.jp}}
}

\begin{document}

\maketitle
\thispagestyle{empty}
\pagestyle{empty}

%%%%%%%%%%%%%%%%%%%%%%%%%%%%%%%%%%%%%%%%%%%%%%%%%%%%%%%%%%%%%%%%%%%%%%%%%%%%%%%%

\begin{abstract}

This letter proposes a formal synthesis of Robust Koopman-Model Predictive Control (RK-MPC), a novel data-driven approach to formal synthesis of systems with nonlinear dynamics. 
We formulate a novel optimization problem for RK-MPC by incorporating specifications described by  Signal Temporal Logic and prove its closed-loop performance. 
Effectiveness of the proposed RK-MPC is evaluated by applying it to the reliable design of an AC-DC power converter. 

\end{abstract}

%%%%%%%%%%%%%%%%%%%%%%%%%%%%%%%%%%%%%%%%%%%%%%%%%%%%%%%%%%%%%%%%%%%%%%%%%%%%%%%%

\section{INTRODUCTION}

Koopman-Model Predictive Control (K-MPC) is a data-driven optimal control method for nonlinear systems \cite{KORDA2018149}. 
The basic idea of K-MPC is to utilize a \textit{linear} or \textit{bilinear} model of a nonlinear plant for predicting the state or output of the nonlinear plant and to generate an optimal input. 
We refer to the linear/bilinear model, guided by the Koopman operator theory \cite{koopman-introduction-springer}, as the \emph{Koopman model}.  
K-MPC has been successfully applied in engineering domains, including robotics \cite{Koopman-operators-in-robot-lerning} and energy systems \cite{developing-koopmanism-in-power-and-energy-systems}: see references therein. 
In \cite{hirose-isie2026}, the authors applied K-MPC to a real AC-DC converter motivated by More Electric Aircraft (MEA) \cite{onboard} and experimentally showed that the K-MPC controller performs better than other existing controllers. 

The so-called Robust K-MPC (RK-MPC), which is a combination of K-MPC with robust control techniques, has been reported; see, e.g., \red{\cite{RK-MPC-ZHANG,RK-MPC-Mamakoukas,RK-MPC-Yan}}. 
The Koopman model involves modeling error arising from the linear/bilinear approximation of a nonlinear plant and from its estimation from finite data. 
This modeling error causes a K-MPC controller to violate constraints prescribed in the optimization problem; see, e.g., \red{\cite{RK-MPC-ZHANG,RK-MPC-Mamakoukas,RK-MPC-Yan}}. 
The RK-MPC controller considers the modeling error of the Koopman model explicitly and hence 
guarantees that the closed-loop response satisfies the prescribed constraints. 
Notably, the authors of \cite{SafEDMD} propose a data-driven modeling framework, called \emph{SafEDMD}, which considers the Koopman model with an error bound proportional to the norm of the state and control input. 
This robustness guarantee makes K-MPC more practical in terms of reliability; see \cite{STRASSER2026101035} for details. 
This reliability has been required by mission-critical applications, including MEA as mentioned before, where the aircraft's electrical standards, MIL-STD-704F \cite{MIL-STD-704F}, specify that, for example, a system's voltage should be maintained above a threshold even during load fluctuation.

In this letter, we propose \emph{formal synthesis of the RK-MPC} and apply it to an AC-DC converter. 
Formal synthesis aims to design systems from temporal logic specifications \cite{FormalMethodsforControl}. 
Temporal logics, including Signal Temporal Logic (STL) \cite{MonitoringTemporalProperties,RobustSatisfactionofTemporal}, are able to describe complex specifications of temporal properties and have been widely used in controller synthesis; see, e.g., \cite{raman2014,sadraddini2015,Miyashita2022}. 
The contributions of this letter are threefold: 
\begin{list}{}{
\setlength{\leftmargin}{1.2em}%
\setlength{\labelwidth}{1em}%
\setlength{\labelsep}{0.2em}%
\setlength{\itemsep}{0pt}%
\setlength{\topsep}{2pt}}
\item[1)] We propose a novel formulation in the RK-MPC that realizes data-driven formal synthesis of controllers for nonlinear systems. 
Following \cite{raman2014}, we encode STL specifications as inequality constraints with max/min functions. 
Then, we formulate an optimization problem with both the inequality constraints from STL and the Koopman model. 
\item[2)] We prove in Theorem~\ref{thm:robust-satisfaction} that the proposed formulation guarantees robust satisfaction of the STL specifications. 
To ensure robustness to the controller, we tighten the inequality constraints associated with the STL specifications by incorporating a modeling error assumed in the SafEDMD framework. 
This implies that the formulation guarantees a robust control strategy that satisfies the STL specifications despite the modeling error. 
\item[3)] We present a case study of reliable design for AC-DC power conversion. 
We set the STL specifications for the AC current and DC voltage of an AC-DC converter and then formulate the optimization problem. 
Through numerical simulation, we show that the AC-DC converter controlled by the proposed K-MPC controller meets the specification despite modeling error. 
\end{list}
Note that the authors of \cite{raman2014} propose formal synthesis of MPC for linear plants by using STL, and the authors of \cite{sadraddini2015} make it robust to disturbances. 
The authors of \cite{Miyashita2022} realize formal synthesis of the K-MPC by using the Koopman model and STL. 
In this letter, we formally synthesize the K-MPC to be robust to modeling error and guarantee robust satisfaction of STL specifications, which, to the best of the authors' knowledge, is the novelty of this letter.

\paragraph*{Notation}
The $1$-norm of a vector $\bm{x}$ is denoted by $\|\bm{x}\|$ and the induced $1$-norm of a matrix $\bm{A}$ is denoted by $\|\bm{A}\|$. 
The notation $\bm{x}[k]$ represents the value of $\bm{x}$ at discrete-time step $k$, $\hat{\bm{x}}[k+\ell|k]$ the $\ell$-step-ahead prediction of $\bm{x}$ evaluated at discrete time $k$, and $\bm{x}[k_1:k_2]$ the sequence $(\bm{x}[k_1], \bm{x}[k_1+1], \dots, \bm{x}[k_2])$.
The symbols $x_i$, $\hat{x}_i$, $u_i$, and $y_i$ represent the $i$th element of $\bm{x}$, $\hat{\bm{x}}$, $\bm{u}$, and $\bm{y}$. 
The symbol $\top$ represents the transpose operation of vectors. 
The symbol $\mathrm{j}=\sqrt{-1}$ stands for the imaginary unit.
The symbols $\Im(\cdot)$ and $\Re(\cdot)$ stand for the imaginary and real part of $(\cdot)$.
%%%%%%%%%%%%%%%%%%%%%%%%%%%%%%%%%%%%%%%%%%%%%%%%%%%%%%%%%%%%
%%%%%%%%%%%%%%%%%%%%%%%%%%%%%%%%%%%%%%%%%%%%%%%%%%%%%%%%%%%%

\section{PRELIMINARIES}
\label{Preliminaries}

\subsection{Koopman Model in SafEDMD \cite{SafEDMD}} \label{Koopman-model-in-SafEDMD}
We present the Koopman model identified by SafEDMD for RK-MPC. 
The Koopman model is a linear/bilinear system on a high-dimensional state space that approximates the original nonlinear dynamics. 
Consider a nonlinear plant represented as a discrete-time controlled system, given by
\begin{equation}
    \label{eq:discrete-nonlinear-system}
    \bm{x}[k+1]=\bm{T}(\bm{x}[k],\bm{u}[k]),
\end{equation}
where $\bm{x}\in\mathbb{R}^{n}$ and $\bm{u}\in\mathbb{R}^{m}$ %(integers $n,m$) 
represent the state and the control input at discrete time $k$ and $\bm{T}:\mathbb{R}^{n}\times\mathbb{R}^{m}\to\mathbb{R}^{n}$ represents a nonlinear map. 
We suppose that the origin is a fixed point under $\bm{u}=\bm{0}$. 
\red{Note that SafEDMD is originally introduced to continuous-time control-affine systems, where a discrete-time estimation problem is presented for data-driven control.}
The Koopman model identified by SafEDMD for the nonlinear plant \eqref{eq:discrete-nonlinear-system} is represented as a bilinear plant with the same input $\bm{u}$, given by 
\begin{equation}
    \label{eq:bilinear-koopman-model}
    \left.
    \begin{aligned}
    \bm{z}[\ell+1] &= \bm{A}\bm{z}[\ell] + \bm{B}_{0}\bm{u}[\ell] + \sum\nolimits^{m}_{i=1}u_{i}[\ell]\bm{B}_{i}\bm{z}[\ell] \\
    \hat{\bm{x}}[\ell+1|k] &= \begin{bmatrix}
        \bm{I} & \bm{0}
    \end{bmatrix} \bm{z}[\ell+1] \quad(\ell=k,k+1,\ldots)
    \end{aligned}
    \right\},
\end{equation}
where $\bm{z}[k] := [\bm{x}[k]^\top, f_{n+1}(\bm{x}[k]), \ldots, f_N(\bm{x}[k])]^\top$, via user-defined functions $f_{n+1},\ldots,f_N$ called \textit{observables}, is a high-dimensional state, $\bm{I}$ the $n\times n$ identity matrix, $\bm{0}$ the $n\times(N-n)$ zero matrix, and $\hat{\bm{x}}$ the estimate of the state $\bm{x}$. 
The finite-dimensional matrices $\bm{A}\in\mathbb{R}^{N\times N}$, $\bm{B}_{0}\in\mathbb{R}^{N\times m}$, and $\bm{B}_{i}\in\mathbb{R}^{N\times N}~(i=1,\ldots,m)$ are estimated from i.i.d samples of the original nonlinear system \eqref{eq:discrete-nonlinear-system}; see \cite{SafEDMD} for the detailed description of the estimation. 
The Koopman model \eqref{eq:bilinear-koopman-model} is used for RK-MPC to predict the state trajectory $\hat{\bm{x}}[k+1:k+N_{\rm P}|k]$ from the current state $\bm{x}[k]$ at time $k$. 
The bilinear model induces a non-convex optimization problem in K-MPC, so we will fix $\bm{z}[\ell]$ of the bilinear term to the constant $\bm{z}[k]$ over the entire prediction horizon, as in \cite{Bruder2021}. 

%Building on \cite{SafEDMD}, we derive a bound on the Koopman model's error, which plays a fundamental role in the RK-MPC. 
\red{SafEDMD derives a bound on the Koopman model's error, which plays a fundamental role in the RK-MPC \cite{SafEDMD}. }
Given the initial state $\bm{x}[k]$ for the prediction, the modeling error is defined as
\begin{equation}
    \bm{e}[k+L|k]:=\hat{\bm{x}}[k+L|k]-\bm{x}[k+L] ~(L=1,2,\ldots).
\end{equation}
\red{To the end of deriving bound for $\bm{e}[k+L|k]$, we make the following assumption.}
\begin{assumption} \label{asmp:projection-error}
    \red{There exists a proportional bound on the error caused by the bilinear approximation \eqref{eq:bilinear-koopman-model} of the nonlinear plant \eqref{eq:discrete-nonlinear-system}.}
\end{assumption}
The following theorem provides a bound on $\bm{e}[k+L|k]$. 
\begin{theorem} \label{thm:error-bound}
    Consider the nonlinear plant \eqref{eq:discrete-nonlinear-system} and \red{the associated Koopman model \eqref{eq:bilinear-koopman-model} which fixes $\bm{z}[\ell]$ of the bilinear term to the constant $\bm{z}[k]$ for K-MPC}. 
    Suppose that \red{Assumption~\ref{asmp:projection-error} holds and} the input $\bm{u}$ is constrained as $|u_{i}|\leq u_{i,\max} ~ (i=1,\ldots,m)$, where $u_{i,\max}$ is a constant. % in an MPC problem. 
    For the prediction of the state $\bm{x}$ from discrete time $k$, the $L$-step prediction error $\bm{e}[k+L|k]$ is bounded with probability $1-\delta$ for $\delta\in(0,1)$ by real-valued constants $c_{\bm{z}}$ and $c_{\bm{u}}$ as follows: 
    \begin{equation}
        \label{eq:multi-step-error}
        \|\bm{e}[k+L|k]\| \leq e_{\max}[k+L|k],
    \end{equation}
    with % as follows: で文章が終わっているのに，ここでwhereは適切ではないのでは？
    $e_{\max}[k+L|k]:= \{ c_{\bm{u}}L\alpha + L\|\bm{z}[k]\|\beta + (c_{\bm{z}}+\beta) (\|\bm{z}[k]\| + L\|\bm{B}_{0}\| \alpha + L \|\bm{z}[k]\| \beta) \sum_{\ell=0}^{L} C(L,\ell+1) \|\bm{A}-\bm{I}\|^{\ell} \} (1 + c_{\bm{z}} + \|\bm{A}-\bm{I}\| + \beta)^L$, 
    $\alpha:= \sum_{i=1}^{m}u_{i,\max}$, and 
    $\beta:= \sum_{i=1}^{m} u_{i,\max} \|\bm{B}_{i}\|$. 
    Note that $C(L, \ell+1)$ represents the combination (binomial coefficient) defined as $C(L, \ell+1) := L!/\{(\ell+1)!(L-\ell-1)!\}$. 
\end{theorem}
\begin{proof}
    In the SafEDMD framework, the one-step error is bounded by 
    $\|\bm{e}[k+1|k]\| \leq  c_{\bm{z}} \|\bm{z}[k]\| + c_{\bm{u}}\|\bm{u}[k]\|$ from Corollary~3.2 of \cite{SafEDMD}. 
    We apply the discrete Gronwall inequality \cite{discrete-gronwall-inequality} to the one-step error bound and immediately obtain \eqref{eq:multi-step-error}.
\end{proof}
\begin{comment}
\begin{remark}
    The constants $c_{\bm{z}}$ and $c_{\bm{u}}$ are tuning parameters for estimating modeling error \cite{SafEDMD} and thus synthesizing an RK-MPC controller below. 
    Larger values of $c_{\bm{z}}$ and $c_{\bm{u}}$ {\color{red}can robustify the RK-MPC controller strongly}, which will be discussed in Sec.~\ref{numerical-simulation}.  
\end{remark}
\end{comment}
%%%%%%%%%%%%%%%%%%%%%%%%%%%%%%%%%%%%%%%%%%%%%%%%%%%%%%%%%%%%

\subsection{Signal Temporal Logic} \label{STL}

STL is a logic to formally specify temporal properties of time-series signals \cite{MonitoringTemporalProperties,RobustSatisfactionofTemporal}. 
We treat the system state trajectory \eqref{eq:discrete-nonlinear-system} as a time-series signal and describe temporal properties of the system's trajectory using STL. 
First, we define the predicate, which is the basic component of STL formulae. 
A predicate $\pi$ for the discrete-time system \eqref{eq:discrete-nonlinear-system} is defined as an inequality of the state $\bm{x}[k]$ and input $\bm{u}[k]$ at time $k$, given by \blue{$\pi := \left\{ \mu(\bm{x}[k],\bm{u}[k]) := \bm{G}_{\bm{x}} \bm{x}[k] + \bm{G}_{\bm{u}} \bm{u}[k] + H\geq 0 \right\}$,}
% \begin{equation*}
%     \pi := \left\{ \mu(\bm{x}[k],\bm{u}[k]) := \bm{G}_{\bm{x}} \bm{x}[k] + \bm{G}_{\bm{u}} \bm{u}[k] + H\geq 0 \right\},
% \end{equation*}
where $\mu: \mathbb{R}^n \times \mathbb{R}^m \to \mathbb{R}$ is the so-called \textit{predicate function} \cite{RobustSatisfactionofTemporal}, $\bm{G}_{\bm{x}}\in\mathbb{R}^{1\times n}$ and $\bm{G}_{\bm{u}}\in\mathbb{R}^{1\times m}$ user-defined row vectors, and $H\in\mathbb{R}$ a user-defined scalar. 
A predicate $\pi$ is satisfied by $\bm{x}[k]$ and $\bm{u}[k]$ at time $k$ if and only if the corresponding predicate function is non-negative, i.e., $\mu(\bm{x}[k],\bm{u}[k])\geq 0$. 
Second, we introduce an STL formula by combining predicates with the following grammar: $\neg$ (negation), $\land$ (\red{conjunction}), $\lor$ (\red{disjunction}), $\mathcal{U}_{[a,b]}$ (timed until), $\mathbf{F}_{[a,b]}$ (timed eventually), and $\mathbf{G}_{[a,b]}$ (timed always). 
%Given a trajectory $(\bm{x}[k+a:k+b],\bm{u}[k+a:k+b])$, the formula $\varphi_1 \,\mathcal{U}_{[a,b]}\, \varphi_2$ holds if and only if $\varphi_{1}$ holds at every time step before $\varphi_{2}$ holds at some time step between $k+a$ and $k+b$. 
See \cite{MonitoringTemporalProperties} for the detailed description of the grammar. 
From \cite{sadraddini2015}, we introduce the horizon length $h^\varphi$ for the STL formula $\varphi$. 
At time $k$, the STL formula $\varphi$ is evaluated on 
the time-sequence of states and inputs $(\bm{x}[k:k+h^\varphi],\bm{u}[k:k+h^\varphi])$.

The validity of an STL formulae is evaluated with an inequality of composition of max and min functions with predicate functions. 
To do so, we introduce the \textit{robustness function} \cite{RobustSatisfactionofTemporal}, which is a function of the predicate functions, denoted as $\rho^\varphi$, for a given STL formula $\varphi$. 
Suppose that $\varphi$ is composed of $p$ predicates given as $\pi_1=\{\mu_1(\bm{x}[k],\bm{u}[k])\geq0\},\ldots,\pi_p=\{\mu_p(\bm{x}[k],\bm{u}[k])\geq0\}$.  
Here, the robustness function $\rho^\varphi$ depends on a set of predicate functions, defined as $M(\bm{x}[k:k+h^\varphi],\bm{u}[k:k+h^\varphi]):=\{\mu_1(\bm{x}[\ell],\bm{u}[\ell]),\ldots,\mu_p(\bm{x}[\ell],\bm{u}[\ell])\}_{\ell=k}^{k+h^\varphi}$, and also depends on time $k$. 
The aforementioned STL grammar excluding negation can be expressed by a max/min function or its combination (see \cite{RobustSatisfactionofTemporal,raman2014}). 
In addition, negation is distributed to predicates, and negation of a predicate is expressed by multiplying the associated predicate function by $-1$ (see \cite{sadraddini2015}). 
%The robustness function of the predicate $\pi$ corresponds to the predicate function $\mu(\bm{x}[k],\bm{u}[k])$. 
%The robustness function of the \textit{and} operator $\varphi=\varphi_1 \land \varphi_2$ and the \textit{until} operator $\varphi=\varphi_1 \mathcal{U}_{[a,b]} \varphi_2$, which we will use in Sec.~\ref{application}, is defined as $\rho^{\varphi}(M(\cdot,\cdot),k) := \min ( \rho^{\varphi_1}(M(\cdot,\cdot),k), \, \rho^{\varphi_2}(M(\cdot,\cdot),k))$ and $\rho^{\varphi}(M(\bm{x}[k+a:k+b],\bm{u}[k+a:k+b]),k) := \max_{k' \in [k+a, k+b]} \min ( \rho^{\varphi_2}(M(\bm{x}[k+a:k+b],\bm{u}[k+a:k+b]),k'), \min_{k'' \in [k, k'-1]} \rho^{\varphi_1}(M(\bm{x}[k+a:k+b],\bm{u}[k+a:k+b]),k'') )$. 
A formula $\varphi$ is satisfied by $\bm{x}[k:k+h^\varphi]$ and $\bm{u}[k:k+h^\varphi]$ at time $k$ if the corresponding robustness function is non-negative, that is, $\rho^\varphi(M(\bm{x}[k:k+h^\varphi],\bm{u}[k:k+h^\varphi]),k) \geq 0$.

%%%%%%%%%%%%%%%%%%%%%%%%%%%%%%%%%%%%%%%%%%%%%%%%%%%%%%%%%%%%%%%%%%%%%%%%%%%%%%%%%%%%%%%%%%%%%%%%%%%%%%%%%%%%%%%%

% \begin{figure}[t]
%     \centering
%     \includegraphics[width=1.0\linewidth]{blockdiagram.pdf}
%     \caption{Conceptual description of formal synthesis of Robust Koopman-MPC. }
%     \label{fig:block-diagram}
% \end{figure}

\section{MAIN RESULT} \label{main-result}

We present the main result of this letter: formal synthesis of the RK-MPC using STL. 
The basic idea is to incorporate inequality constraints for a robustness function, corresponding to a given STL specification, into the optimization problem. 
The K-MPC controller proposed in this letter solves at each time step $k$ of the closed-loop operation the following optimization problem:
\begin{mini!}|s|[3]
    {\substack{\bm{u}^{*}[k:k+h_{\rm P}-1]}}
    {J\left( \hat{\bm{x}}[k:k+h_{\rm P}|k], \bm{u}^{*}[k:k+h_{\rm P}] \right) \label{eq:kmpcstl1-cost}}
    {\label{eq:kmpcstl1}}
    {}
    \addConstraint{\left.\begin{aligned} 
    \bm{z}[\ell+1] &= \bm{A}\bm{z}[\ell]+\bm{B}_{0}\bm{u}^{*}[\ell]+\sum^{m}_{i=1}u^{*}_{i}[\ell]\bm{B}_{i}\bm{z}[k] \\
    \hat{\bm{x}}[\ell+1|k] &= \bm{C}\bm{z}[\ell+1] \quad
    (\ell = k, \ldots, k+h_{\rm P}-1)
    \end{aligned} \right\} \label{eq:kmpcstl1-predictor}}
    \addConstraint{\bm{z}[k]} {=\begin{bmatrix} 
    \bm{x}[k]^\top & f_{n+1}(\bm{x}[k]) & \cdots & f_{N}(\bm{x}[k]) \end{bmatrix}^\top \label{eq:kmpcstl1-initial-condition}}
    \addConstraint{|u^*_i[\ell]|} {\leq u_{i,\max} ~ (\red{i=1,\ldots,m}, ~ \ell=k,\ldots,k+h_{\rm P}) \label{eq:kmpcstl1-input-const}}
    \addConstraint{\left.\begin{aligned}
    & \rho^\varphi(M',\ell) \geq 0 \quad(\ell = k-h^{\varphi}+1,\ldots, k) \\ 
    & M' = M(\bm{x}[\ell:k],\bm{u}[\ell:k]) \cup \\
    & M^{\tighten}(\hat{\bm{x}}[k+1:\ell+h^\varphi|k],\bm{u}^*[k+1:\ell+h^\varphi]) \\
    \end{aligned} \right\} \label{eq:kmpcstl1-stl-past}}
    \addConstraint{\left.\begin{aligned}
    & \rho^\varphi(M',\ell) \geq 0 \quad(\ell = k+1,\ldots, k+h_{\rm S}) \\ 
    & M' = M^{\tighten}(\hat{\bm{x}}[\ell:\ell+h^\varphi|k],\bm{u}^*[\ell:\ell+h^\varphi]) \\
    \end{aligned} \right\}. \label{eq:kmpcstl1-stl-future}}
\end{mini!}
The symbol $h_{\rm P}$ is the prediction horizon for the state $\bm{x}$, and $J$ is the cost function. 
At time $k$, the decision variable $\bm{u}^{*}[k]$ is applied to the plant. 
The Koopman model \eqref{eq:bilinear-koopman-model} is introduced as the equality constraints \eqref{eq:kmpcstl1-predictor}. 
The initial value of $\bm{z}$ is set as \eqref{eq:kmpcstl1-initial-condition}. 
The input constraint is introduced as \eqref{eq:kmpcstl1-input-const}. 
Eqs.~\eqref{eq:kmpcstl1-cost}, \eqref{eq:kmpcstl1-predictor}, \eqref{eq:kmpcstl1-initial-condition}, \eqref{eq:kmpcstl1-input-const} are the same as those for the K-MPC \cite{KORDA2018149}. 
We impose the STL formula $\varphi$ on the optimization problem by the inequality constraints \eqref{eq:kmpcstl1-stl-past} and \eqref{eq:kmpcstl1-stl-future}. 
%As mentioned in Sec.~\ref{STL}, the STL formula $\varphi$ at time $k$ is evaluated on the time-sequence of states/inputs $(\bm{x}[k:k+h^\varphi],\bm{u}[k:k+h^\varphi])$. 
Eq.~\eqref{eq:kmpcstl1-stl-past} handles the sequence of past states/inputs and future states/inputs, and Eq.~\eqref{eq:kmpcstl1-stl-future} does only the sequence of future states/inputs. 
As a novel point of this letter, 
we formulate the STL constraints \eqref{eq:kmpcstl1-stl-past} and \eqref{eq:kmpcstl1-stl-future} from \textit{tightened} predicate functions. 
The \textit{tightened} predicate functions corresponding to the predicates $\{\pi_i\}_{i=1}^p$ can be introduced as 
\begin{equation}
    \label{eq:tightened-predicate}
    \mu_{i}^{\tighten}(\hat{\bm{x}}[\ell|k],\bm{u}[\ell]) := \mu_{i}(\hat{\bm{x}}[\ell|k],\bm{u}[\ell]) - \|\bm{G}_{\bm{x}}\| e_{\max}[\ell|k].
\end{equation}
The term $\|\bm{G}_{\bm{x}}\| e_{\max}[\ell|k]$ is the compensation for the maximum error associated with $\bm{G}_{\bm{x}} \hat{\bm{x}}[\ell|k]$ in terms of the predicate function $\mu_i$. 
We also introduce a set of the tightened predicate functions as $M^{\tighten}(\hat{\bm{x}}[k+1:k+h|k],\bm{u}[k+1:k+h]):=\{ \mu_{1}^{\tighten}(\hat{\bm{x}}[\ell|k],\bm{u}[\ell]), \ldots , \mu_{p}^{\tighten}(\hat{\bm{x}}[\ell|k],\bm{u}[\ell]) \}_{\ell=k+1}^{k+h}$ for time length $h$. 
Then, we formulate the \textit{tightened} robustness function $\rho^\varphi(M',\ell)$ by composing max/min functions with the tightened predicate functions in $M^{\tighten}$ as described in Sec.~\ref{STL}. 
%Since the reformulated robustness function is a combination of max/min functions with respect to predicate functions, it is monotonically non-decreasing with respect to predicate functions. 
%Thus, the inequality constraint for the robustness function reformulated from the tightened predicate functions is also \emph{tightened} {\color{red}(robustness is guaranteed)}. 
The inequality constraints for the tightened robustness function are \emph{tightened} compared to those for the robustness function from the original predicate functions.

Now, we prove that if \eqref{eq:kmpcstl1} is feasible, then the robust closed-loop performance for the specification $\varphi$ is guaranteed. 
The following lemma on the boundedness of the tightened predicate function \eqref{eq:tightened-predicate} is the first step to the main result. 

\begin{lemma} \label{thm:tightened-predicate-bound}
    The tightened predicate function \eqref{eq:tightened-predicate} serves as a lower bound for the original predicate function, i.e., for $i=1,\ldots,p$ and $\ell=k+1,k+2,\ldots,$ 
    \begin{equation}
        \label{eq:predicate-inequality}
        \mu_{i}^{\tighten}(\hat{\bm{x}}[\ell|k], \bm{u}[\ell]) \le \mu_{i}(\bm{x}[\ell], \bm{u}[\ell]),
    \end{equation}
    where the right-hand side corresponds to the predicate function evaluated at the true trajectory.
\end{lemma}
\begin{proof}
    It holds $\mu_{i}(\bm{x}[\ell], \bm{u}[\ell]) - \mu_{i}^{\tighten}(\hat{\bm{x}}[\ell|k], \bm{u}[\ell]) 
    = \bm{G}_{\bm{x}} (\bm{x}[\ell] - \hat{\bm{x}}[\ell|k]) + \|\bm{G}_{\bm{x}}\| e_{\max}[\ell|k]$ from \eqref{eq:tightened-predicate}. 
    It is straightforward that $\bm{G}_{\bm{x}} (\bm{x}[\ell] - \hat{\bm{x}}[\ell|k]) \geq -\|\bm{G}_{\bm{x}}\| \|\bm{x}[\ell] - \hat{\bm{x}}[\ell|k]\| = -\|\bm{G}_{\bm{x}}\| \|\bm{e}[\ell|k]\|$ from the basic properties of the induced norms. 
    Since $e[\ell|k]$ is bounded by $e_{\max}[\ell|k]$ as \eqref{eq:multi-step-error}, it holds $-\|\bm{G}\| \|\bm{e}[\ell|k]\| + \|\bm{G}_{\bm{x}}\|e_{\max}[\ell|k] \geq 0$. 
    Thus, \eqref{eq:predicate-inequality} is valid.
\end{proof}
%As the second step, we guarantee that the tightened robustness function $\rho^\varphi(M',\ell)$, where $M'$ stands for the set of the predicate functions, is also tightened. 
The following lemma on the boundedness of the tightened robustness function $\rho^\varphi(M',\ell)$ in \eqref{eq:kmpcstl1-stl-past} and \eqref{eq:kmpcstl1-stl-future} is the second step. 
\begin{lemma}
    \label{lemma:2}
    The robustness function $\rho^\varphi(M',\ell)$ using the tightened predicate functions \eqref{eq:tightened-predicate} serves as a lower bound for the original robustness function, i.e., for $\ell=k+1,k+2,\ldots,$ 
    \begin{equation}
        \label{eq:robustness-inequality}
        \rho^{\varphi}(M',\ell) \leq \rho^{\varphi}(M(\bm{x}[\ell:\ell+h^\varphi],\bm{u}[\ell:\ell+h^\varphi]),\ell),
    \end{equation}
    where again the right-hand side corresponds to the robustness function evaluated at the true trajectory.
\end{lemma}
\begin{proof}
    The robustness function is a non-decreasing function for the predicate function since the robustness function is composed of max/min functions. Therefore, \eqref{eq:robustness-inequality} follows from \eqref{eq:predicate-inequality} in Lemma~\ref{thm:tightened-predicate-bound}. 
\end{proof}
Finally, we state the main theorem showing that the proposed formulation \eqref{eq:kmpcstl1} guarantees robustness. 
\begin{theorem} \label{thm:robust-satisfaction} 
    Suppose that \red{the RK-MPC optimization problem \eqref{eq:kmpcstl1} is feasible.} 
    Then, the closed-loop response of the nonlinear plant \eqref{eq:discrete-nonlinear-system} satisfies the STL formula $\varphi$ \red{with probability $1-\delta$}. 
\end{theorem}
\begin{proof}
    The solution of \eqref{eq:kmpcstl1} forces $\rho^{\varphi}(M',\ell)\geq0$. Therefore, $\rho^{\varphi}(M(\bm{x}[\ell:\ell+h^\varphi],\bm{u}[\ell:\ell+h^\varphi]),\ell)\geq0$ follows from \eqref{eq:robustness-inequality} in Lemma~\ref{lemma:2}.
\end{proof}
\begin{remark}
% 結果の新規性とutility（どうsynthesisに繋がるのか）をクリアに説明する。
% 仮定の妥当性---解の存在--についてコメントする。仮定しっぱなしは良くない理論の典型...
    Theorem~\ref{thm:robust-satisfaction} guarantees that the proposed formulation \eqref{eq:kmpcstl1} is robust to the modeling error of the Koopman model while satisfying the given STL formula $\varphi$. 
    The novelty of Theorem~\ref{thm:robust-satisfaction} is to prove the robust satisfaction of the STL formula. 
\end{remark}
\begin{remark}
    The optimization problem \eqref{eq:kmpcstl1} is a mixed-integer problem because the max and min functions in the STL constraints \eqref{eq:kmpcstl1-stl-past} and \eqref{eq:kmpcstl1-stl-future} are expressed with binary variables \cite{raman2014}. 
    Furthermore, \eqref{eq:kmpcstl1} can be relaxed by rewriting hard constraints \eqref{eq:kmpcstl1-stl-past} and \eqref{eq:kmpcstl1-stl-future} into soft constraints if the solution of \eqref{eq:kmpcstl1} is not located \cite{sadraddini2015}. 
    \red{Then, the satisfaction of the STL specification is not necessarily guaranteed at the cost of feasibility. The limitation of our RK-MPC is the recursive feasibility, which is an open problem of MPC with STL specifications \cite{sadraddini2015}.}
\end{remark}
% \begin{remark}
%     \red{While our MPC formulation fixes the Koopman model for prediction, the authors of \cite{RK-MPC-Yan} propose the RK-MPC without STL specifications but with online model update, which guarantees recursive feasibility. Therein, error bounds for the current Koopman model are also updated using feedback of a current state. It is useful to incorporate online update of the Koopman model into our formulation.} 
% \end{remark}

%%%%%%%%%%%%%%%%%%%%%%%%%%%%%%%%%%%%%%%%%%%%%%%%%%%%%%%%%%%%%%%%
%%%%%%%%%%%%%%%%%%%%%%%%%%%%%%%%%%%%%%%%%%%%%%%%%%%%%%%%%%%%%%%%

\section{APPLICATION}  \label{application}
In this section, we apply the proposed formulation to conduct the RK-MPC of an AC-DC converter. 
The plant's setting is based on our previous report \cite{hirose-isie2026}, and the parameters used in this letter are mainly from \cite{modeling-and-impedance}, which presents an experiment of an electric power system for MEA.

\begin{figure}[t]
    \vspace{3.7pt}
    \centering
    \includegraphics[width=0.95\linewidth]{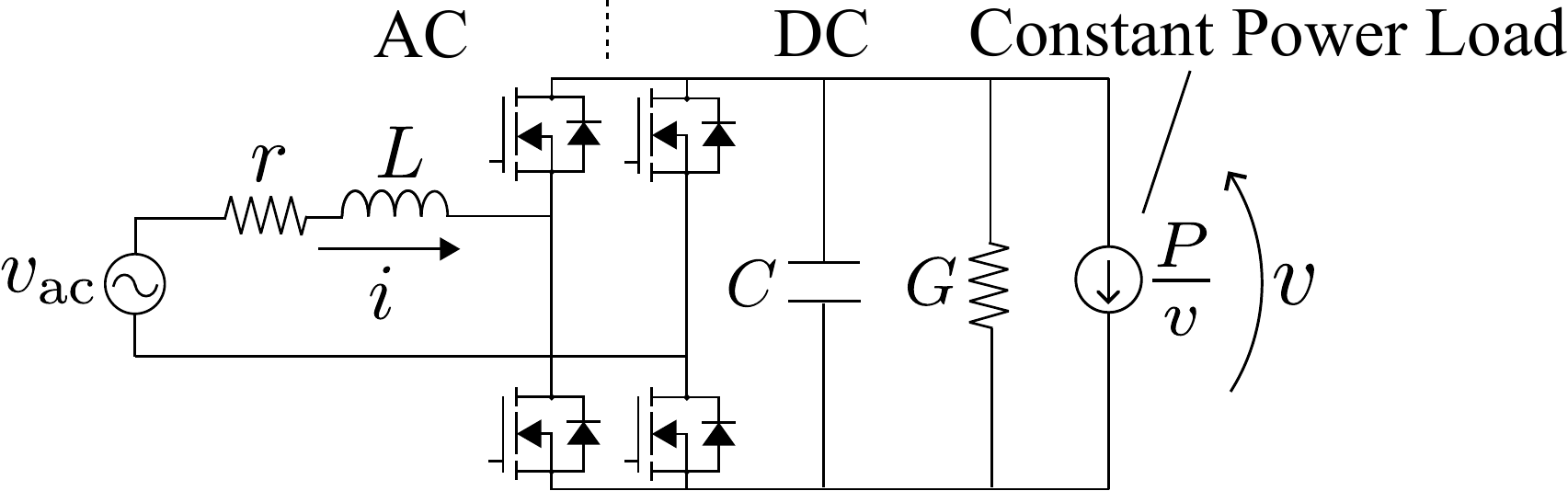}
    \caption{Basic configuration of an AC-DC converter. The Constant Power Load (CPL) is represented by the symbol of a current source.}
    \label{fig:AC-DC-converter}
\end{figure}

\subsection{AC-DC Converter} \label{AC-DC-Converter}

Fig.~\ref{fig:AC-DC-converter} shows a single-phase AC-DC converter containing a full-bridge boost rectifier. 
The state variables are defined as $\bm{x}(t)=[i(t),v(t)]^\top$, where $i(t)$ is the AC current and $v(t)$ is the DC voltage. 
The input variable is the duty ratio $s(t)\in[-1,1]$. 
The AC supply voltage $v_{\rm ac}(t)$ in continuous time $t$ is represented as $v_{\rm ac}(t) = E\sin(\omega t)$ with the amplitude $E=160/\sqrt{3}\,\rm V$ and the angular frequency $\omega=2\pi\cdot400\,\rm rad/s$ (period $T=2.5\,\rm ms$). 
The DC loads consist of a resistive load of $G=1/47\,\rm S$ and a Constant Power Load (CPL) with $P=1.5\,\rm kW$.  
The current of the CPL is represented as $i_{\rm CPL}(t) =P/v(t)$. 
The inductance is $L=20\,\mu\rm H$, the capacitance $C=1.2\,\rm mF$, and the resistance $r=0.2\,\Omega$. 
The dynamical model of the AC-DC converter is represented by the nonlinear control-affine system as follows:
\begin{equation}
    \label{eq:AC-DC-converter1}
    \left.
    \begin{alignedat}{2}
        -L\mathrm{d}i(t) / \mathrm{d}t &= s(t)v(t) + ri(t) - E\sin(\omega t) \\
        C\mathrm{d}v(t) / \mathrm{d}t &= s(t)i(t) - Gv(t) - {P}/{v(t)} 
    \end{alignedat}
    \right\}.
\end{equation}

The control objectives for the AC-DC converter are listed as follows: 1) The mean value of the DC voltage $v(t)$ is equal to the desired value $V_{\rm d}=270\,\rm V$; 2) The power factor of the converter is equal to one; precisely, $i(t)$ becomes $I_{\rm d}\sin\omega t$ with constant amplitude $I_{\rm d}$. 
Here, $I_{\rm d}$ is calculated as $I_{\rm d}=79.9\,\rm A$ from a power-balance equation between the AC source and the DC loads (see \cite{hirose-isie2026} for the detailed calculation).

We impose the following specifications to ensure reliable control of the AC-DC converter. 
First, the short time average of the DC voltage $v(t)$ over the AC period $T$ must be maintained at or above $250\,\rm V$. This is set in accordance with MIL-STD-704F \cite{MIL-STD-704F}. 
Second, any excursion of the AC current $i(t)$ above the rating of $82\,\rm A$ must last less than $5\,\rm ms$, namely two AC cycles. 
If $i(t)$ remains above the rated value for the pre-defined duration, the AC source will stop supplying power, which is generally called \textit{trip} of the AC source. 
A voltage sag on the DC side can cause this trip; the AC source temporarily increases its output current to keep the output power constant, resulting in a possible current excursion. 
Overall, the above specifications are to prevent the trip of the AC source while maintaining the DC voltage. 
The specifications will be described with STL in Sec.~\ref{setting-of-controller}. 

\begin{figure*}
    \vspace{3.4pt}
    \centering
    \includegraphics[width=0.96\linewidth]{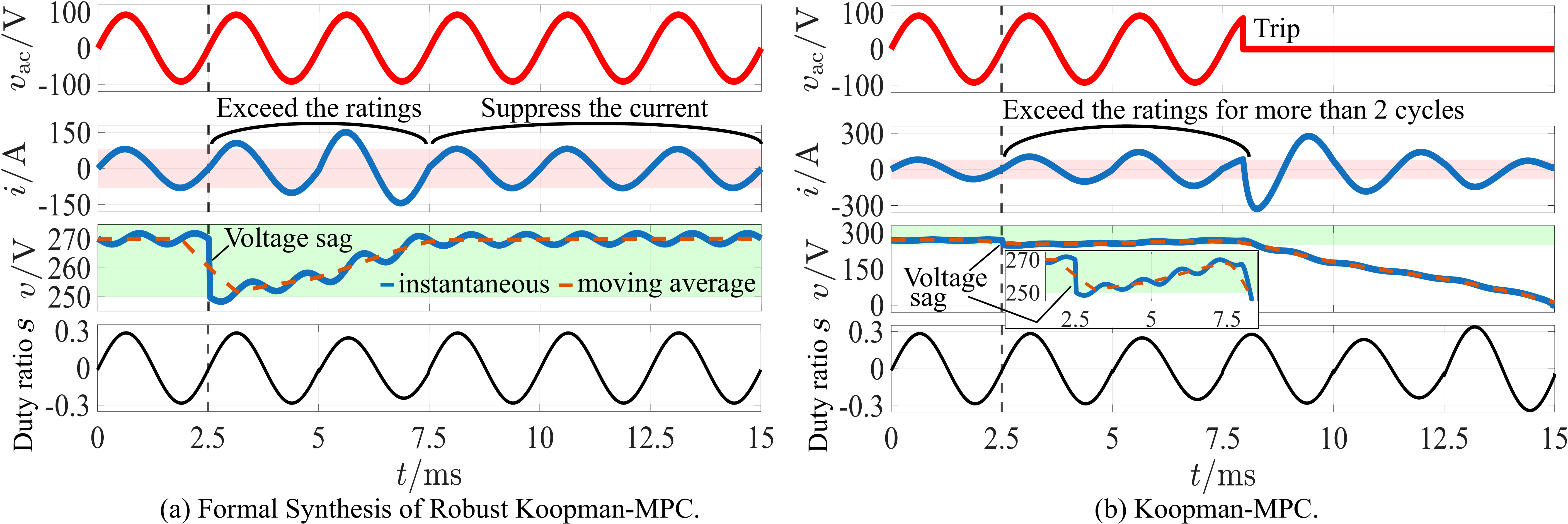}
    \caption{Comparison of the control results by Formal Synthesis of Robust Koopman-MPC and Koopman-MPC. Graphs in row 1 show the AC source voltage, graphs in row 2 the AC current (\emph{blue} line) and the ratings $|i|\leq82\,\rm A$ (\emph{red} area), graphs in row 3 show the DC voltage (\emph{blue} line), its moving average (\emph{orange} line), and the range $v\geq250\,\rm V$ (\emph{green} area), and graphs in row 4 the duty ratio. The vertical \emph{dotted} line represents the moment when the voltage sag occurs. }
    \label{fig:control-result}
\end{figure*}

%%%%%%%%%%%%%%%%%%%%%%%%%%%%%%%%%%%%%%%%%%%%%%%%%%%%%%%%%%%%
%%%%%%%%%%%%%%%%%%%%%%%%%%%%%%%%%%%%%%%%%%%%%%%%%%%%%%%%%%%%

\subsection{Koopman Modeling} \label{Koopman-model-for-ac-dc-converter}

First, we introduce a set of observables for the AC-DC converter from \cite{hirose-isie2026}. 
To address the control objectives of the AC-DC converter, the authors propose in \cite{hirose-isie2026} using a dynamical map, specifically Generalized State-Space Averaging (GSSA), to synthesize observables. 
\red{Note that the choice of observables enables a successful performance of Koopman-MPC in an experimental circuit \cite{hirose-isie2026}.}
The index-$h\in\{0,1\}$ GSSA for state $x(t)$ is introduced as \blue{$\langle x \rangle_{h}(t) := \int_{t-T}^{t} x(\tau)\exp({-\mathrm{j}\omega h\tau}){\rm d}\tau/T$,}
% \begin{equation}
%     \label{eq:GSSA-definition}
%     \langle x \rangle_{h}(t) := \frac{1}{T} \int_{t-T}^{t} x(\tau)\exp({-\mathrm{j}\omega h\tau}){\rm d}\tau,
% \end{equation}
where $T$ is the period of the AC voltage. %$\mathrm{j}=\sqrt{-1}$ stands for the imaginary unit.
That is, $\langle x \rangle_h(t)$ at time $t$ is the moving harmonic-average of the state $x(t)$ over $[t-T,t]$. %during the time between $t-T$ and $t$.
The $\langle i \rangle_{1}$ is related to the phasor of the AC current $i(t)$ and $\langle v \rangle_{0}$ is the mean value of $v(t)$. 
The reference signals of $i(t)$ and $v(t)$ are encoded with the GSSA as $\langle I_{\rm d}\sin(\omega t)\rangle_{1}(t)=-\mathrm{j}I_{\rm d}/2$ and $\langle V_{\rm d}\rangle_{0}(t)=V_{\rm d}$, which become time-invariant and will be used for the RK-MPC. 
In this letter, we compute \red{GSSA} at every AC period $T$ and label $\langle x \rangle_{h}(kT)$ as $\langle x \rangle_{h}[k]$ by an integer $k$. 
This integer $k$ is regarded as discrete time. 
Note that SafEDMD assumes that the equilibrium point of the plant is the origin \cite{SafEDMD}, but the target state of our converter system is not the origin because the reference signals of the GSSA variables $\langle i \rangle_1(t)$ and $\langle v \rangle_0(t)$ are non-zero. 
Thus, we take a new state $\bm{y}=[y_1,y_2,y_3]^\top:=[\Im\langle i \rangle_{1}-(-I_{\rm d}/2),\Re\langle i \rangle_{1},\langle v \rangle_{0}-V_{\rm d}]^\top$. 
When $\langle i \rangle_1[k]$ and $\langle v \rangle_0[k]$ reach their references in sufficient time, $\bm{y}[k]$ reaches the origin. 
By using the output $\bm{y}$, following \cite{hirose-isie2026}, we synthesize a set of observables as $\bm{z} := [y_1,y_2,y_3,1/(y_3+V_{\rm d})-1/V_{\rm d}]^\top$. 
The fourth element of $\bm{z}$ is motivated by the presence of the nonlinear term $P/v$ in \eqref{eq:AC-DC-converter1}.

Second, we introduce the input variable at discrete time $k$ determined by the RK-MPC controller. 
Following \cite{hirose-isie2026}, we represent the duty ratio as a sine wave, given by $s(t)=\tilde{u}_{1}\sin(\omega t) + \tilde{u}_{2}\cos(\omega t)$, 
where $\tilde{u}_{1},\tilde{u}_{2}$ are constants at steady states. 
On the basis of \cite{ijcta}, the steady-state values of $\tilde{u}_{1},\tilde{u}_{2}$ are calculated from the references of $\langle i \rangle_{1}$ and $\langle v \rangle_{0}$ as follows: $\bar{u}_{1} = -2(GV_{\rm d}+P/V_{\rm d})/I_{\rm d}$ and $\bar{u}_{2} = \omega LI_{\rm d}/V_{\rm d}$. 
We then introduce the control variable $\bm{u}=[u_1,u_2]^\top$ for the Koopman model \eqref{eq:bilinear-koopman-model} and set $\tilde{u}_1=\bar{u}_1+u_1$ and $\tilde{u}_2=\bar{u}_2+u_2$. 
In this letter, we regard $u_{1},u_{2}$ as the control variables that actuate the GSSA variables $\langle i \rangle_1 [k]$ and $\langle v \rangle_0 [k]$. 
That is, their values at discrete time $k$ (namely, $\bm{u}[k]=[u_{1}[k],u_{2}[k]]^\top$) regulate the values of the lifted variable $\bm{z}$ at time $k+1$ ($\bm{z}[k+1]$) as in \eqref{eq:bilinear-koopman-model}. 
This implies that $s(t)$ is updated at every period of the AC voltage. 

To obtain training datasets, we compute the state trajectory $\bm{x}(t)=[i(t),v(t)]^\top ~ (t\in[0,2T])$. 
Initial points $\bm{x}^{(d=1)}(t=0)=[i^{(d=1)}(t=0), v^{(d=1)}(t=0)]^\top, \ldots, \bm{x}^{(d=D)}(t=0)=[i^{(d=D)}(t=0), v^{(d=D)}(t=0)]^\top$ are generated with uniform random numbers, and one of three types of the control variable, $\bm{u}=[0,0]^\top,[0.01,0]^\top,[0,0.01]^\top$, is set  for each of the generated initial points. 
Here, $D$ is the total number of the generated trajectories, denoted as $\bm{x}^{(d=1)}(t)=[i^{(d=1)}(t),v^{(d=1)}(t)]^\top, \ldots, \bm{x}^{(d=D)}(t)=[i^{(d=D)}(t),v^{(d=D)}(t)]^\top ~ (t\in[0,2T])$ are computed. 
The training datasets of one-step maps are then collected as $\{(\bm{y}^{(d)}[k=1],\bm{u}^{(d)},\bm{y}^{(d)}[k=2])\}^{D}_{d=1}$ where $\bm{y}^{(d)}[k=1]$ and $\bm{y}^{(d)}[k=2]$ are computed via GSSA from the associated state trajectory $\bm{x}(t)~(t\in[0,2T])$. % by \eqref{eq:GSSA-definition}. 
Following \cite{SafEDMD}, we identify three Koopman models \eqref{eq:bilinear-koopman-model} from the datasets with $D=15,90,300$. 
%The error rates in predictions by these Koopman models are computed as $\mathrm{Error~rate} := \sqrt{ \sum_{i=1}^{3} \sum_{k=2}^{40} ( y_{i}[k] - \hat{y}_{i}[k|1] )^2} / \sqrt{ \sum_{i=1}^{3} \sum_{k=2}^{40} y_{i}[k]^2}$. 
%The error rates for the three Koopman models with $D=15,90,300$ are $9.6\times10^{-3}$, $7.4\times10^{-3}$, and $7.1\times10^{-3}$. 
These models have different degrees of modeling error and will be used to verify the robust performance of the proposed control. %in order to verify that our proposed method is robust to different degrees of modeling error. 

%%%%%%%%%%%%%%%%%%%%%%%%%%%%%%%%%%%%%%%%%%%%%%%%%%%%%%%%%%%%%%%%%%%%%%%%%%%%%%%%%%%%%%%%%%%%%%%%%%%%%%%%%%%%%%%%%%%%%%%%

\subsection{Setting of Controller Synthesis} \label{setting-of-controller}
We illustrate the setting of controller synthesis applied to the AC-DC converter. 
The cost function is set as $J(\hat{\bm{y}}[k:k+h_{\rm P}|k], \bm{u}^{*}[k:k+h_{\rm P}]) = \sum_{\ell=k}^{k+h_{\rm P}-1} \left\{ \hat{\bm{y}}[\ell+1|k]^{\top}\bm{Q}\hat{\bm{y}}[\ell+1|k] + \varDelta\bm{u}^{*}[\ell]^{\top} \bm{R} \varDelta\bm{u}^{*}[\ell]  \right\}$, 
% \begin{equation*}
%     \sum_{\ell=k}^{k+h_{\rm P}-1} 
%     \left\{ \hat{\bm{y}}[\ell+1|k]^{\top}\bm{Q}\hat{\bm{y}}[\ell+1|k] + \varDelta\bm{u}^{*}[\ell]^{\top} \bm{R} \varDelta\bm{u}^{*}[\ell]  \right\},
% \end{equation*}
where $\bm{Q}=\mathrm{diag}[0,1,5]$ is the state weight matrix, $\bm{R}=\mathrm{diag}[0.5,0.5]$ is the input weight matrix, and $\varDelta \bm{u}^*[\ell] = \bm{u}^*[\ell+1]-\bm{u}^*[\ell]$ is the increment of the input. 
The cost function aims to penalize the deviation of the GSSA variables $\bm{y}$ from the origin $\bm{0}$ and the increment of the input variables, latter of which is motivated by minimizing the fluctuation of the input time-series in steady state. 
We set $\bm{Q}=\mathrm{diag}([0,1,5])$ because focusing only on $\Re\langle i \rangle_{1}$ and $\langle v \rangle_{0}$ is sufficient to achieving the control objectives stated in Sec.~\ref{AC-DC-Converter}. 
Since $y_1$ includes $I_{\rm d}$, $y_1$ is associated with the amplitude of the AC current, and the control objective 2) is not concerned with the amplitude but with the phase of the AC current. 
Therefore, we focus on $y_2$ rather than $y_1$. 
The prediction horizon is set as \red{$h_{\rm P}=5~(12.5\,\rm ms)$} based on MIL-STD-704F \cite{MIL-STD-704F}, where the DC voltage is stipulated to restore in the order of several \red{tens} of milliseconds. 
The input $\bm{u}$ is constrained in the range of $[-0.01,0.01]$. 
We represent the control specification described in Sec.~\ref{AC-DC-Converter} by the following STL formula: \blue{$\varphi = \left(\langle v \rangle_{0}\geq 250\right) ~ \mathcal{U}_{[0,2]} ~ \{ (\Im\langle i \rangle_{1}\geq -40.6) \land (\Re\langle i \rangle_{1}\geq -5.8) \land (\Re\langle i \rangle_{1}\leq 5.8)\}.$}
% \begin{equation}
%     \label{eq:stl-spec}
%     \begin{aligned}
%         \varphi = 
%         &\left(\langle v \rangle_{0}\geq 250\right) ~ \mathcal{U}_{[0,2]} ~ 
%         \{ (\Im\langle i \rangle_{1}\geq -40.6) \\ 
%         &\land (\Re\langle i \rangle_{1}\geq -5.8) \land (\Re\langle i \rangle_{1}\leq 5.8)\}.
%     \end{aligned}
% \end{equation}
The duration $[0,2]$ is described in discrete time with sampling period $T$. 
The STL specification is satisfied in $[k,k+2]$ for every discrete time $k=1,2,\ldots$ 
The predicate $\langle v \rangle_{0}\geq 250$ describes that the time averaging of $v$ is larger than and equal to $250\,\rm V$. 
The predicates $\{ (\Im\langle i \rangle_{1}\geq -40.6) \land (\Re\langle i \rangle_{1}\geq -5.8) \land (\Re\langle i \rangle_{1}\leq 5.8)\}$ are associated with the current limit. 
When the AC current takes the maximum amplitude and the minimum power factor $0.99$ i.e., $i(t)=(82\,{\rm A})\sin(\omega t+\theta)$ with $\cos\theta=0.99$, we have $\langle i \rangle_{1}=\pm 5.8 - \mathrm{j} 40.6$, from which the above predicates are derived.

%%%%%%%%%%%%%%%%%%%%%%%%%%%%%%%%%%%%%%%%%%%%%%%%%%%%%%%%%%%%%%%%%%%%%%%%%%%%%%%%%%%%%%%%%%%%%%%%%%%%%%%%%%%%%%%%%%%%%%%%

\subsection{Numerical Results} \label{numerical-simulation}

Fig.~\ref{fig:control-result} shows numerical results of the closed-loop performance of the AC-DC converter. 
%Here, we use the Koopman model identified from $D=300$ datasets and set the tuning parameters $c_{\bm{z}}$ and $c_{\bm{u}}$ of $e_{\max}$ in \eqref{eq:multi-step-error} as $c=c_{\bm{z}}=c_{\bm{u}}=0.005$. 
\red{Here, we use the Koopman model identified from $D=300$ datasets and set the constants $c_{\bm{z}}$ and $c_{\bm{u}}$ of $e_{\max}$ in \eqref{eq:multi-step-error} as $c=c_{\bm{z}}=c_{\bm{u}}=0.005$. 
The selection of $c$ is based on trial and error. 
That is, we determined their values by checking whether the MPC controller numerically locates a solution of the optimization problem and the solution guarantees the satisfaction of the STL specification.}
For comparison, we implement the standard K-MPC controller by excluding the inequalities of the robustness function. 
Control simulations were conducted with numerical integration of \eqref{eq:AC-DC-converter1}. 
%The duty ratio $s(t)$ was implemented in continuous time through a zeroth-order hold with a sampling period $200\,\mu\rm s$. 
The optimization problem was solved by using the MIQP solver, Gurobi 13.0.1. 
The average computation time per iteration was $338\,\rm ms$ on a desktop with $2.4\,\rm GHz$ Intel Xeon. 
A voltage sag is modeled by subtracting $20\,\rm V$ from the DC voltage $v(t)$ at $t=2.5\,\rm ms$, denoted by the dotted vertical lines. 
Regarding control objective 1), both the controllers achieve small steady-state errors for the mean value of $v(t)$ in $[0,2.5]\,\rm ms$ before the sag. 
Regarding control objective 2), we compute the power factor during the steady-state $[0,2.5]\,\rm ms$ by using the Hilbert transform. 
The value of power factor is more than $0.99$ for both the controllers, implying that they achieve the control objectives well before the sag. 
After the voltage sag, both the controllers maintain the DC voltage above $250\,\rm V$ and restore it to the reference value $270\,\rm V$. 
To restore the voltage, the AC source increases its output current $i(t)$. 
The current exceeds the rated value $82\,\rm A$ during the transient response for both the controllers. 
The transient responses differ between the two controllers. 
For the standard K-MPC controller in Fig.~\ref{fig:control-result}(b), the overcurrent persists for three cycles and finally leads to the trip of power supply, for which the AC voltage $v_{\rm ac}(t)$ is enforced to zero. 
Without power supply from the AC source, the DC load voltage $v(t)$ inevitably goes to $0\,\rm V$.
On the other hand, for the proposed RK-MPC controller in Fig.~\ref{fig:control-result}(a), the current returns below its rated value within the predefined duration of two cycles, and therefore no trip of power supply occurs. 
This comparison indicates that the proposed controller achieves reliable performance of the AC-DC converter, which is never achieved by the K-MPC controller.

We show the robustness test for the proposed formulation. 
As mentioned in Sec.~\ref{Koopman-model-for-ac-dc-converter}, we derived the three Koopman models with different degrees of modeling error. 
\red{The Koopman models identified from fewer than 15 datasets fail to achieve the control objectives.}
Here, we check if the proposed formulation is robust to varying degrees of modeling error, namely, using the datasets with different $D$ for the modeling. 
%We also check how the tuning parameters $c=c_{\bm{z}}=c_{\bm{u}}$ in $e_{\max}$ of \eqref{eq:multi-step-error} affect the control performance. 
Specifically, we test this by designing the controllers in Sec.~\ref{setting-of-controller} using the three Koopman models and applying them in the fault scenario described in the previous paragraph. 
TABLE~\ref{tab:evaluation} shows the test results. 
The controller with $c=0.005$, which we utilize in the previous paragraph, guarantees the robust satisfaction of the specification despite degrees of modeling error. 
Furthermore, we change \red{the constants $c=c_{\bm{z}}=c_{\bm{u}}$ of $e_{\max}$ in \eqref{eq:multi-step-error}} and check how these \red{constants} affect the control performance. 
%The small $c$ assumes low degree of modeling error and, as a result, degrades the robustness of the proposed formulation. 
In the case of $c=0$, the maximum error $e_{\max}$ is assumed to be $e_{\max}=0$, thus the controller does not take into account the modeling error. 
\red{This result makes the benefit of the robustness margin clear.}
Consequently, the controllers with any model fail to achieve the specification (violated). 
In the case of $c=0.003$, the maximum error $e_{\max}$ is assumed to be too small, and the controllers also fail to achieve the specification. 
On the other hand, all of the controllers with $c\geq0.01$ fail to locate a solution of \eqref{eq:kmpcstl1}. 
Thus, we show that the robustness and feasibility of the proposed formulation is affected by the \red{constant} $c$ and, in this case study, ensured as long as we choose the \red{constant} $c$ as $c=0.005$.

\begin{table}
    \vspace{3.4pt}
    \centering
    \caption{Robustness test of the proposed Robust Koopman-MPC controller.}
    \begin{tabular}{c||cccc}
        \makecell{Amount of data \\ for modeling} & $c=0$ & $c=0.003$ & $c=0.005$ & $c=0.01$ \\ \hline
        15 samples & Violated & Violated & Satisfied & Infeasible \\
        90 samples & Violated & Violated & Satisfied & Infeasible \\
        300 samples & Violated & Violated & Satisfied & Infeasible \\
    \end{tabular}
    \label{tab:evaluation}% c=0.15 all are infeasible
\end{table}
\begin{comment}
\begin{table}[t]
    \centering
    \caption{Robustness test of the proposed RK-MPC controller.}
    \begin{tabular}{c|ccc}
        \makecell{Amount of data \\ for modeling} & 15 samples & 90 samples & 300 samples \\ 
        \hline\hline
        Control result & Satisfied & Satisfied & Satisfied \\
    \end{tabular}
    \label{tab:evaluation}% c=0.15 all are infeasible
\end{table}
\end{comment}

%%%%%%%%%%%%%%%%%%%%%%%%%%%%%%%%%%%%%%%%%%%%%%%%%%%%%%%%%%%%%%%%%%%%%%%%%%%%%%%%%%%%%%%%%%%%%%%%%%%%%%%%%%%%%%%%%%%%%%%%
\raggedbottom

\section{CONCLUSIONS}

This letter proposed the formal synthesis of RK-MPC and applied it to the AC-DC power converter. 
Our proposed formulation handles the STL specification, including the description of temporal properties and theoretically guarantees the robust closed-loop performance in terms of satisfying the STL specification. 
We showed the effectiveness of formal synthesis of RK-MPC for the AC-DC converter by comparison with the standard K-MPC. 
Furthermore, we verified the robustness of the proposed formulation using the Koopman models with varying degrees of modeling error. 
Future work includes \red{the guarantee of recursive feasibility and the incorporation of the Koopman model with online model update.} 
%an experimental realization of the formal synthesis of RK-MPC for the AC-DC converter.

\addtolength{\textheight}{-12cm}   % This command serves to balance the column lengths

%%%%%%%%%%%%%%%%%%%%%%%%%%%%%%%%%%%%%%%%%%%%%%%%%%%%%%%%%%%%%%%%%%%%%%%%%%%%%%%%

%%%%%%%%%%%%%%%%%%%%%%%%%%%%%%%%%%%%%%%%%%%%%%%%%%%%%%%%%%%%%%%%%%%%%%%%%%%%%%%%

%%%%%%%%%%%%%%%%%%%%%%%%%%%%%%%%%%%%%%%%%%%%%%%%%%%%%%%%%%%%%%%%%%%%%%%%%%%%%%%%
\begin{comment}
\section*{APPENDIX}

Appendixes should appear before the acknowledgment.
\end{comment}

%%%%%%%%%%%%%%%%%%%%%%%%%%%%%%%%%%%%%%%%%%%%%%%%%%%%%%%%%%%%%%%%%%%%%%%%%%%%%%%%

% \bibliographystyle{ieeetr}
% \bibliography{Sankou}

\end{document}